\documentclass[10pt,journal,a4paper]{IEEEtran}
\usepackage[normalem]{ulem}
\usepackage{multirow}
\usepackage{caption}
\usepackage{subcaption}
\usepackage{caption}
\usepackage{graphicx}
\usepackage{longtable}
\usepackage{courier}
\usepackage{float}
\usepackage{cite}
\usepackage{mathtools}
\usepackage{amsmath,amssymb}
\usepackage{bbold}
\usepackage{dsfont}
\usepackage[in]{fullpage}
\usepackage[verbose]{wrapfig}
\usepackage{amsthm}
\usepackage{authblk}
\usepackage{xcolor}

\usepackage[linesnumbered]{algorithm2e}

\newcommand{\dfn}{\stackrel{\triangle}{=}}
\newcommand{\RomanNumeralCaps}[1]
    {\MakeUppercase{\romannumeral #1}}

\def\vv{\boldsymbol{v}}

\def\dv{\boldsymbol{d}}

\def\Uc{\mathcal{U}}
\def\Lc{\boldsymbol{\mathcal{L}}}
\def\Zc{\mathcal{Z}}
\def\Tau{\mathcal{T}}
\def\Tau{\mathcal{T}}

\def\dvlambda{\boldsymbol{d_\lambda}}

\newcommand\smallV{
  \mathchoice
    {{\scriptstyle\mathcal{V}}}
    {{\scriptstyle\mathcal{V}}}
    {{\scriptscriptstyle\mathcal{V}}}
    {\scalebox{.7}{$\scriptscriptstyle\mathcal{O}$}}
  }

\newcommand{\defeq}{\triangleq}

\usepackage[T1]{fontenc}
\usepackage[utf8]{inputenc}
\usepackage{authblk}
\usepackage[margin=0.75in]{geometry}
\usepackage{chngcntr}
\newtheorem{theorem}{Theorem}
\newtheorem{lemma}{Lemma}
\newtheorem{corollary}{Corollary}

\newtheorem{remark}{Remark}

\title{Coded Caching with Shared Caches: \\Fundamental Limits with Uncoded Prefetching}

\author{Emanuele Parrinello, Ay\c{s}e \"{U}nsal and Petros~Elia
\thanks{The authors are with the Communication Systems Department at EURECOM, Sophia Antipolis, 06410, France (email: parrinel@eurecom.fr, unsal@eurecom.fr, elia@eurecom.fr). The work is supported by the European Research Council under the EU Horizon 2020 research and innovation program / ERC grant agreement no. 725929. (ERC project DUALITY)}
\thanks{This work is to appear in part in the proceedings of ITW 2018\nocite{MN14,WanTP15,YuMA16}. An extended version of this work can be found in \cite{PEU_arxiv_multi}.}
}

\begin{document}
\thispagestyle{plain}
\maketitle

\begin{abstract}
The work identifies the fundamental limits of coded caching when the $K$ receiving users share $\Lambda\leq K$ helper-caches, each assisting an arbitrary number of different users. The main result is the derivation of the exact optimal worst-case delivery time --- under the assumption of uncoded cache placement --- for any user-to-cache association profile where each such profile describes how many users are helped by each cache. This is achieved with a new information-theoretic converse that is based on index coding and which proves that a simple XOR-shrinking-and-removal coded-caching scheme is optimal irrespective of the user-to-cache association profile. All the results also apply directly to the related coded caching problem with multiple file requests.
\end{abstract}


\section{Introduction and System Model\label{sec:systemModel}}
In this work, we consider a basic broadcast configuration where a transmitting server has access to a library of $N$ files $W^{1},\dots,W^{N}$, each of size equal to one unit of `file', and where this transmitter wishes to communicate some of these files via a shared (bottleneck) broadcast link, to $K$ receiving users, each having access to one of $\Lambda\leq K$ helper nodes that will serve as caches of content from the library.
The communication process is split into $a)$ the cache-placement phase, $b)$ the user-to-cache assignment phase during which each user is assigned to a single cache, and $c)$ the delivery phase during which each user requests a single file independently and during which the transmitter aims to deliver these requested files, taking into consideration the cached content and the user-to-cache association.

\paragraph{Cache placement phase}
During this phase, helper nodes store content from the library without having knowledge of the users' requests. Each helper cache has size $M\leq N$ units of file, and no coding is applied to the content stored at the helper caches; this corresponds to the common case of \textit{uncoded cache placement}. The cache-placement algorithm is oblivious of the subsequent user-to-cache association $\mathcal{U}$.

\paragraph{User-to-cache association}
After the caches are filled, each user is assigned to exactly \emph{one} helper node/cache, from which it can download content at zero cost. Specifically, each cache $
\lambda = 1,2,\dots, \Lambda$, is assigned to a set of users $\mathcal{U}_\lambda$, and all these disjoint sets \[\mathcal{U}\dfn \{\mathcal{U}_1,\mathcal{U}_2,\dots ,\mathcal{U}_\Lambda\}\] form a partition of the set of users $\{1,2,\dots,K\}$, describing the overall association of the users to the caches.

This cache assignment is independent of the cache content and independent of the file requests to follow. We here consider any arbitrary user-to-cache association $\mathcal{U}$, thus allowing the results to reflect both an ability to choose/design the association, as well as to reflect possible association restrictions due to randomness or topology. Similarly, having the user-to-cache association being independent of the requested files, is meant to reflect the fact that such associations may not be able to vary as quickly as a user changes the requested content.

\paragraph{Content delivery}
The delivery phase commences when each user $k = 1,\dots,K$ requests from the transmitter, any \emph{one} file $W^{d_{k}}$, $d_{k}\in\{1,\dots,N\}$ out of the $N$ library files.
Upon notification of the entire \emph{demand vector} $\boldsymbol{d}=(d_1,d_2,\cdots,d_{K})\in\{1,\dots,N\}^K$, the transmitter aims to deliver the requested files, each to their intended receiver, and the objective is to design a \emph{caching and delivery scheme $\chi$} that does so with limited (delivery phase) duration $T$. The delivery algorithm is aware of $\mathcal{U}$.

\begin{figure}[t!]
\centering
\includegraphics[width=0.6\linewidth]{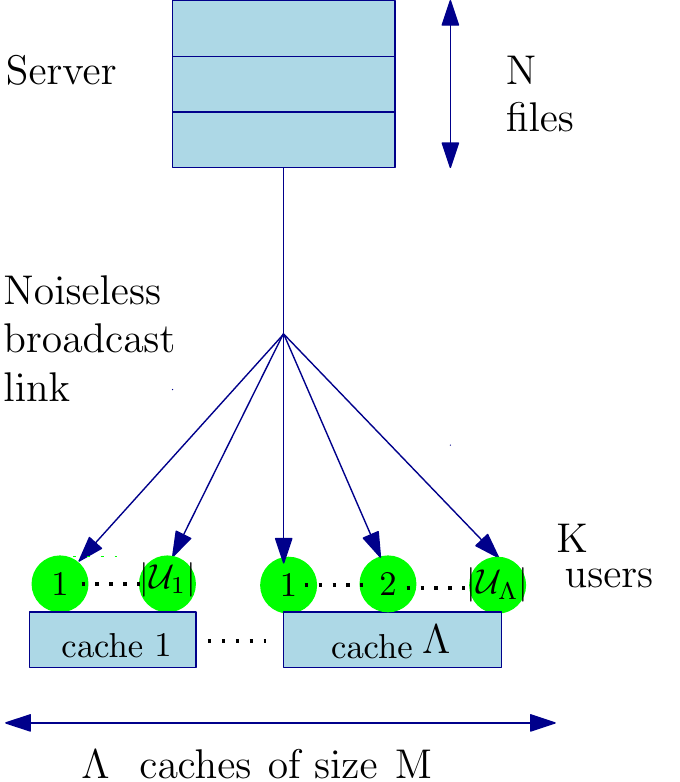}
    \label{system_pic}
  \caption{Shared-link broadcast channel with shared caches.}
\end{figure}

\paragraph{Performance measure}
As one can imagine, some user-to-cache association instances $\mathcal{U}$ allow for higher performance than others; for instance, one can suspect that more uniform profiles may be preferable. Part of the objective of this work is to explore the effect of such associations on the overall performance. Toward this, for any given $\mathcal{U}$, we consider the association \textit{profile} (sorted histogram)
$$\Lc=(\mathcal{L}_{1},\dots,\mathcal{L}_{\Lambda})$$ where $\mathcal{L}_{\lambda}$ is the number of users assigned to the $\lambda$-th \emph{most populated} helper node/cache\footnote{Here $\Lc$ is simply the vector of the cardinalities of $\mathcal{U}_\lambda, ~\forall\lambda\in\{1,\dots,\Lambda\}$, sorted in descending order. For example, $\mathcal{L}_{1}=6$ states that the highest number of users served by a single cache, is $6$.}. Naturally, $\sum_{\lambda=1}^\Lambda \mathcal{L}_{\lambda} = K$. Each profile $\Lc$ defines a class $\mathcal{U}_{\Lc}$ comprising all the user-to-cache associations $\mathcal{U}$ that share the same profile $\Lc$.

As in \cite{MN14}, $T$ is the number of time slots, per file served per user, needed to complete delivery of any file-request vector\footnote{The time scale is normalized such that one time slot corresponds to the optimal amount of time needed to send a single file from the transmitter to the receiver, had there been no caching and no interference.} $\dv$. We use $T(\mathcal{U},\dv,\chi)$ to define the delay required by some generic caching-and-delivery scheme $\chi$ to satisfy demand $\dv$ in the presence of a user-to-cache association described by $\mathcal{U}$. To capture the effect of the user-to-cache association, we will characterize the optimal worst-case delivery time associated to each class
\begin{equation}
T^*(\Lc)\defeq \min_{\chi} \max_{(\mathcal{U},\dv) \in (\mathcal{U}_{\Lc},\{1,\dots,N\}^K)} T(\mathcal{U},\dv,\chi).\label{eq:T*_def}
\end{equation}
Our interest is in the regime of $N\geq K$ where there are more files than users.

\paragraph{Context and related work}
Our work can be seen as an extension of the work in~\cite{MN14} which considered the uniform setting where $\Lambda=K$ (where $\Lc=(1,1,\dots,1)$), and which provided the breakthrough of coded caching that allowed for a worst-case delivery time of $T = \frac{K(1-\frac{M}{N})}{1+K\frac{M}{N}}$. The concept of coded caching has been adapted to a variety of settings, in different works that include \cite{ZhangE16b,BidokhtiWT16isit,GhorbelKY:16,ZE:17tit,SenguptaTS15,CaoTXL16,RoigTG17a,CaoTaoMultiAntenna18,PiovClerckISIT18} and many others. \nocite{lampiris2018lowCSIT}
Interestingly, under the assumption of uncoded cache placement where caches store uncoded subfiles, this performance --- for the case where each user has its own cache --- was proven in~\cite{WanTP15} (see also \cite{YuMA16}) to be exactly optimal.

The setting of coded caching with shared caches, was explored in \cite{MND13}, as well as in~\cite{Diggavi_IT} which considered a similar shared-cache setting as here --- under a uniform user-to-cache association where each cache serves an equal number of users --- and which proposed a coded caching scheme that was shown to perform to within a certain constant factor from the optimal.

In this context of coded caching with shared caches, we here explore the effect of user-to-cache association profiles, and how profile skewness affects performance. This aspect is crisply revealed here as a result of a novel scheme and an outer bound that jointly provide exact optimality results. This direction is motivated by the realistic constraints in assigning users to caches, where these constraints may be due to topology, cache capacity, and other factors.


\paragraph*{Paper outline}
The main results are presented in Section~\ref{sec:results}. Section \ref{sec:converse} presents the information-theoretic converse, while Section~\ref{sec:scheme} describes the coded caching scheme and presents an example. Finally Section~\ref{sec:discussion} draws some basic conclusions based on the obtained results.

\subsection{Notation}
We will use $\gamma \defeq \frac{M}{N} $ to denote the normalized cache size. We denote the cache content at helper node $\lambda = 1,2,\dots,\Lambda$ by $\Zc_\lambda$.
For $n$ denoting a positive integer, $[n]$ refers to the following set $[n]\triangleq \{1,2,\dots,n\}$, and $2^{[n]}$ denotes the power set of $[n]$. The expression $\alpha | \beta$ denotes that integer $\alpha$ divides  integer $\beta$. Permutation and binomial coefficients are denoted and defined by $P(n,k)\defeq  \frac{n!}{(n-k)!}$ and $\binom{n}{k}\defeq  \frac{n!}{(n-k)!k!}$, respectively.
For a set $\mathcal{A}$, $|\mathcal{A}|$ denotes its cardinality.
$\mathbb{N}$ represents the natural numbers. We denote the lower convex envelope of the points $\{(i, f(i)) | i \in [n]\cup \{0\}\}$ for some $n\in \mathbb{N}$ by $Conv(f(i))$. 
For $n\in \mathbb{N}$, we denote the symmetric group of all permutations of $[n]$ by $S_n$.
To simplify notation, we will also use such permutations $\pi\in S_n$ on vectors $\vv \in \mathbb{R}^n$, where $\pi(\vv)$ will now represent the action of the permutation matrix defined by $\pi$, meaning that the first element of $\pi(\vv)$ is $\vv_{\pi(1)}$ (the $\pi(1)$ entry of $\vv$), the second is $\vv_{\pi(2)}$, and so on. Similarly $\pi^{-1}(\cdot )$ will represent the inverse such function and $\pi_s(\vv)$ will denote the sorted version of a real vector $\vv$ in descending order.


\section{Main Results \label{sec:results}}

We present our main result in the following theorem.
\begin{theorem}\label{thm:PerClassSingleAntenna}
In the $K$-user shared-link broadcast channel with $\Lambda$ shared caches of normalized size $\gamma$, the optimal delivery time within any class/profile $\Lc$ is
\begin{equation}\label{eq:TS_L}
T^*(\Lc)=Conv\bigg(\frac{\sum_{r=1}^{\Lambda-\Lambda\gamma}\mathcal{L}_r{\Lambda-r\choose \Lambda\gamma}}{{\Lambda\choose \Lambda\gamma}}\bigg)
\end{equation}
at points $\gamma\in \{\frac{1}{\Lambda},\frac{2}{\Lambda},\dots,1\}$.
\end{theorem}
The converse and achievability of \eqref{eq:TS_L} are proved in Section~\ref{sec:converse} and Section~\ref{sec:scheme}, respectively.

\begin{remark}
We note that the converse that supports Theorem~\ref{thm:PerClassSingleAntenna}, encompasses the class of all caching-and-delivery schemes $\chi$ that employ uncoded cache placement under a general sum cache constraint
$\frac{1}{\Lambda}\sum_{\lambda=1}^\Lambda |\mathcal{Z}_\lambda | = M$
which does not \emph{necessarily} impose an individual cache size constraint. The converse also encompasses all scenarios that involve a library of size $\sum_{n\in[N]}|W^{n}| = N$ but where the files may be of different size. In the end, even though the designed optimal scheme will consider an individual cache size $M$ and equal file sizes, the converse guarantees that there cannot exist a scheme (even in settings with uneven cache sizes or uneven file sizes) that exceeds the optimal performance identified here.
\end{remark}

From Theorem~\ref{thm:PerClassSingleAntenna}, we see that in the uniform case\footnote{Here, this uniform case, naturally implies that $\Lambda|K$.} where $\Lc=(\frac{K}{\Lambda},\frac{K}{\Lambda},\dots,\frac{K}{\Lambda})$, the expression in~\eqref{eq:TS_L} reduces to
\[T^*(\Lc)=\frac{K(1-\gamma)}{\Lambda\gamma+1}\]
matching the achievable delay presented in \cite{MND13}, which was recently proved in \cite{WeUlukus17} --- in the context of the multiple file requests problem --- to be optimal under the assumption of uncoded cache placement.

The following corollary relates to this uniform case.
\begin{corollary}\label{cor:ressym}
In the uniform user-to-cache association case where $\Lc=(\frac{K}{\Lambda},\frac{K}{\Lambda},\dots,\frac{K}{\Lambda})$, the aforementioned optimal delay $T^*(\Lc)=\frac{K(1-\gamma)}{\Lambda\gamma+1}$ is smaller than the corresponding delay $T^*(\Lc)$ for any other non-uniform class.
\end{corollary}

\begin{proof}
The proof that the uniform profile induces the smallest delay among all profiles, follows directly from the fact that in~\eqref{eq:TS_L}, both $\mathcal{L}_r$ and ${\Lambda-r\choose \Lambda\gamma}$ are non-increasing with $r$.\end{proof}

In a nutshell, what Theorem~\ref{thm:PerClassSingleAntenna} and Corollary~\ref{cor:ressym} reveal is that profile non-uniformities always bring about increased delays, and the more skewed the profile is, the larger is the delay. This is reflected in Figure~\ref{fig:performance} which shows --- for a setting with $K=30$ users and $\Lambda=6$ caches --- the memory-delay trade-off curves for different user-to-cache association profiles. As expected, Figure~\ref{fig:performance} demonstrates that when all users are connected to the same helper cache, the only gain arising from caching is the well known \textit{local caching gain}. On the other hand, when users are distributed uniformly among the caches (i.e., when $\mathcal{L}_{\lambda}=\frac{K}{\Lambda},\forall\lambda\in[\Lambda]$) the caching gain is maximized and the delay is minimized.

\begin{figure}[t!]
\centering
\includegraphics[width=1\linewidth]{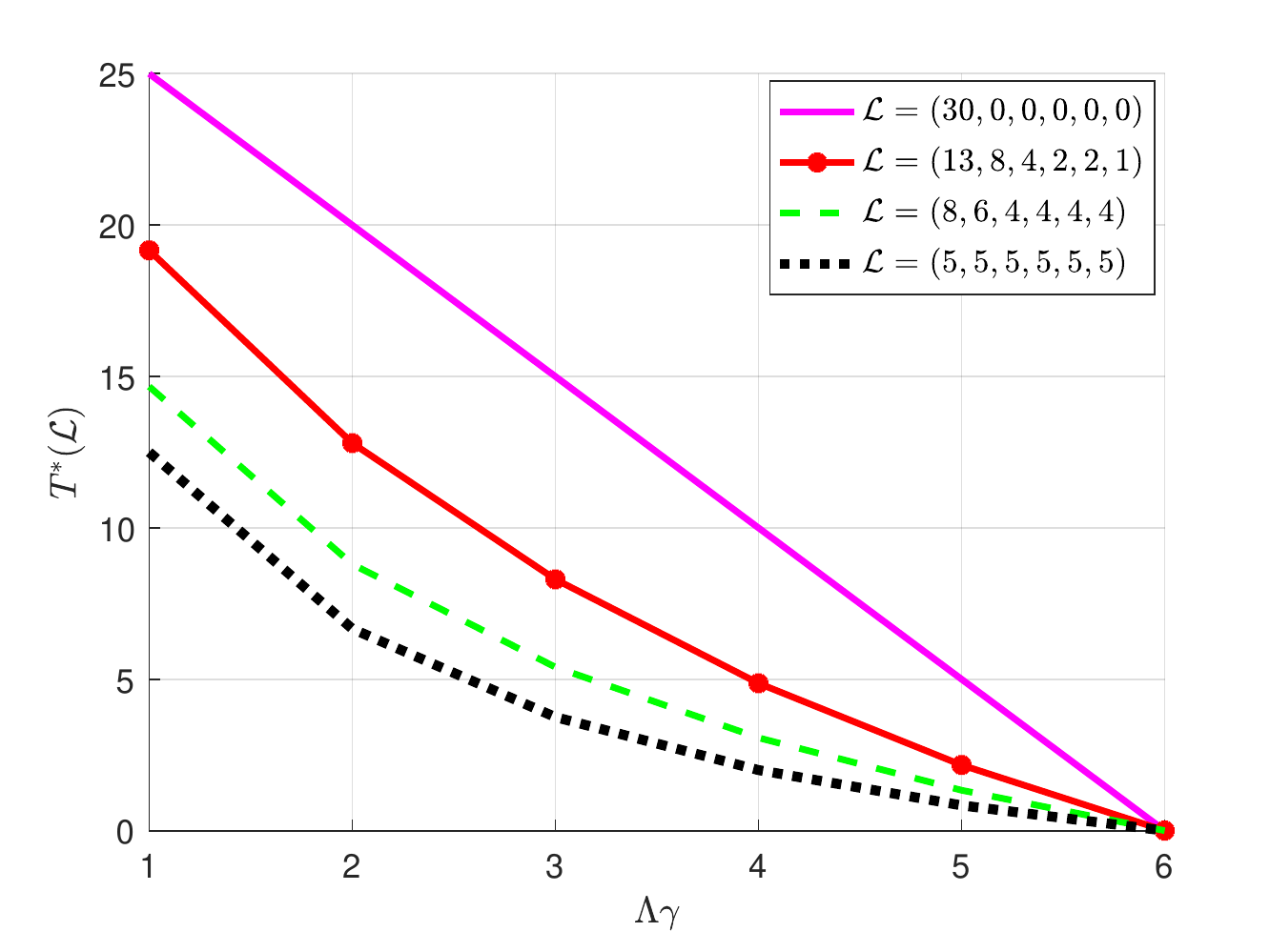}
  \caption{Optimal delay for different user-to-cache association profiles $\Lc$, for $K=30$ users and $\Lambda=6$ caches.}
  \label{fig:performance}
\end{figure}

\begin{remark}[Shared-link coded caching with multiple file requests]\label{rem:multipleFilerequstsResult}
In the error-free shared-link case ($N_0 = 1$), with file-independence and worst-case demand assumptions, the shared-cache problem here is closely related to the coded caching problem with multiple file requests per user, where now $\Lambda$ users with their own cache, request in total $K\geq \Lambda$ files. In particular, changing a bit the format, now each demand vector $\dv = (d_1,d_2,\dots,d_K)$ would represent the vector of the indices of the $K$ requested files, and each user $\lambda = \{1,2,\dots,\Lambda\}$, would request those files from this vector $\dv$, whose indices\footnote{For example, having $\mathcal{U}_2 = \{3,5,7\}$, means that user 2 has requested files $W^{d_{3}},W^{d_{5}},W^{d_{7}}$.} form the set $\mathcal{U}_\lambda \subset [K]$.
At this point, as before, the problem is now defined by the user-to-file association $\mathcal{U} = \{\mathcal{U}_1,\mathcal{U}_2,\dots ,\mathcal{U}_\Lambda\}$ which describes --- given a fixed demand vector $\dv$ --- the files requested by any user. From this point on, the equivalence with the original shared cache problem is complete. As before, each such $\mathcal{U}$ again has a corresponding (sorted) profile $\Lc=(\mathcal{L}_{1},\mathcal{L}_{2},\dots,\mathcal{L}_{\Lambda})$, and belongs to a class $\mathcal{U}_{\Lc}$ with all other associations $\mathcal{U}$ that share the same profile $\Lc$. As we quickly show in the extended version \cite[Appendix Section \RomanNumeralCaps{7}-H]{PEU_arxiv_multi}, our scheme and converse can be adapted to the multiple file request problem, and thus directly from Theorem~\ref{cor:ressym} we conclude that for this multiple file request problem, the optimal delay $T^*(\Lc)\defeq \min_{\chi} \max_{(\mathcal{U},\dv) \in (\mathcal{U}_{\Lc},\{1,\dots,N\}^K)} T(\mathcal{U},\dv,\chi)$ corresponding to any user-to-file association profile $\Lc$, takes the form $T^*(\Lc)= Conv\bigg(\frac{\sum_{r=1}^{\Lambda-\Lambda\gamma}\mathcal{L}_{r}{\Lambda-r\choose \Lambda\gamma}}{{\Lambda\choose \Lambda\gamma}}\bigg)$.  At this point we close the parenthesis regarding multiple file requests, and we refocus exclusively on the problem of shared caches.
\end{remark}


\section{Information Theoretic Converse\label{sec:converse}}
Toward proving Theorem~\ref{thm:PerClassSingleAntenna}, we develop a lower bound on the normalized delivery time in (\ref{eq:T*_def}) for each given user-to-cache association profile $\Lc$. The proof technique is based on the breakthrough in~\cite{WanTP15} which --- for the case of $\Lambda = K$, where each user has their own cache --- employed index coding to bound the performance of coded caching. Part of the challenge here will be to account for having shared caches, and mainly to adapt the index coding approach to reflect non-uniform user-to-cache association classes.

We will begin with lower bounding the normalized delivery time $T(\mathcal{U},\dv,\chi)$, for any user-to-cache association $\mathcal{U}$, demand vector $\dv$ and a generic caching-delivery strategy $\chi$.

\paragraph*{Identifying the distinct problems} 
The caching problem is defined when the user-to-cache association $\mathcal{U}=\{\mathcal{U}_\lambda \}_{\lambda=1}^\Lambda$ and demand vector $\dv$ are revealed.
What we can easily see is that there are many combinations of $\{\mathcal{U}_\lambda \}_{\lambda=1}^\Lambda$ and $\dv$ that yield the same coded caching problem. After all, any permutation of the file indices requested by users assigned to the same cache, will effectively result in the same coded caching problem.
As one can see, every \emph{distinct} coded caching problem is fully defined by $\{\dvlambda\}_{\lambda=1}^\Lambda$, where $\dvlambda$ denotes the vector of file indices requested by the users in $\mathcal{U}_\lambda$, i.e., requested by the $|\mathcal{U}_\lambda|$ users associated to cache $\lambda$. The analysis is facilitated by reordering the demand vector $\boldsymbol{d}$ to take the form $\boldsymbol{d}(\Uc)\dfn (\boldsymbol{d_1}, \cdots, \boldsymbol{d_\Lambda})$.
Based on this, we define the set of worst-case demands associated to a given profile $\Lc$, to be
$$\mathcal{D}_{\Lc} = \{\dv(\mathcal{U}): \dv\in \mathcal{D}_{wc}, \mathcal{U} \in \mathcal{U}_{\Lc}
\}$$ where $\mathcal{D}_{wc}$ is the set of worst-case demand vectors $\dv$ whose entries are different (i.e., where $d_i \neq d_j, ~i,j\in[\Lambda],~i\neq j$). We will convert each such coded caching problem into an index coding problem.

\paragraph*{The corresponding index coding problem} 
To make the transition to the index coding problem, each requested file $W^{\dvlambda(j)}$ is split into $2^\Lambda$ disjoint subfiles $W^{\dvlambda(j)}_\Tau,\Tau\in 2^{[\Lambda]}$ where $\Tau\subset[\Lambda]$ indicates the set of helper nodes in which $W^{\dvlambda(j)}_\Tau$ is cached\footnote{Notice that by considering a subpacketization based on the power set $2^{[\Lambda]}$, and by allowing for any possible size of these subfiles, the generality of the result is preserved. Naturally, this does not impose any sub-packetization related performance issues because this is done only for the purpose of creating a converse.}. Then --- in the context of index coding --- each subfile $W^{\dvlambda(j)}_\Tau$ can be seen as being requested by a different user that has as side information all the content $\Zc_\lambda$ of the same helper node $\lambda$. Naturally, no subfile of the form $W^{\dvlambda(j)}_\Tau, \; \forall~\Tau \ni\lambda$ is requested, because helper node $\lambda$ already has it. Therefore the corresponding index coding problem is defined by $K2^{\Lambda-1}$ requested subfiles, and it is fully represented by the side-information graph $\mathcal{G}=(\mathcal{V}_{\mathcal{G}},\mathcal{E}_{\mathcal{G}})$, where $\mathcal{V}_{\mathcal{G}}$ is the set of vertices (each vertex/node representing a different subfile $W^{\dvlambda(j)}_\Tau, \Tau\not\ni\lambda$) and $\mathcal{E}_{\mathcal{G}}$ is the set of direct edges of the graph. Following standard practice in index coding, a directed edge from node $W^{\dvlambda(j)}_\Tau$ to $W^{\boldsymbol{d_{\lambda'}}(j')}_{\Tau'}$ exists if and only if $\lambda'\in\Tau$.  For any given $\mathcal{U}$, $\boldsymbol{d}$ (and of course, for any scheme $\chi$) the total delay $T$ required for this index coding problem, is the completion time for the corresponding coded caching problem.

\paragraph*{Lower bounding $T(\Uc,\dv,\chi)$}
We are interested in lower bounding $T(\Uc,\dv,\chi)$ which represents the total delay required to serve the users for the index coding problem corresponding to the side-information graph $\mathcal{G}_{\Uc,\dv}$ defined by $\Uc,\dv,\chi$ or equivalently by $\boldsymbol{d}(\Uc),\chi$.

In the next lemma, we remind the reader --- in the context of our setting --- the useful index-coding converse from~\cite{ArbaWang13}.

\begin{lemma}(Cut-set-type converse \cite{ArbaWang13})\label{cor_dof}
For a given $\Uc,\dv,\chi$, in the corresponding side information graph $\mathcal{G}_{\Uc,\dv}=(\mathcal{V}_{\mathcal{G}},\mathcal{E}_{\mathcal{G}})$ of the shared-link broadcast channel with $\mathcal{V}_{\mathcal{G}}$ vertices/nodes and $\mathcal{E}_{\mathcal{G}}$ edges, the following inequality holds
\begin{equation}\label{eq:indexbound}T\geq \sum_{\smallV \in \mathcal{V_{J}}}|\smallV|
\end{equation}
for every acyclic induced subgraph $\mathcal{J}$ of $\mathcal{G}_{\Uc,\dv}$, where $\mathcal{V}_{\mathcal{J}}$ denotes the set of nodes of the subgraph $\mathcal{J}$, and where $|\smallV|$ is the size of the message/subfile/node $\smallV$.
\end{lemma}

\paragraph*{Creating large acyclic subgraphs}
Lemma~\ref{cor_dof} suggests the need to create (preferably large) acyclic subgraphs of $\mathcal{G}_{\mathcal{U},\dv}$. The following lemma describes how to properly choose a set of nodes to form a large acyclic subgraph.

\begin{lemma}\label{lem:cons_acyclic}
An acyclic subgraph $\mathcal{J}$ of $\mathcal{G}_{\Uc,\dv}$ corresponding to the index coding problem defined by $\Uc,\dv,\chi$ for any $\Uc$ with profile $\Lc$, is designed here to consist of all subfiles $W^{\boldsymbol{d_{\sigma_{s}(\lambda)}}(j)}_{\Tau_{\lambda}},~\forall j\in [\mathcal{L}_{\lambda}],~\forall \lambda\in [\Lambda]$ for all $\Tau_{\lambda}\subseteq [\Lambda]\setminus \{\sigma_s(1),\dots,\sigma_s(\lambda)\}$ where $\sigma_s\in S_{\Lambda}$ is the permutation such that $|\mathcal{U}_{\sigma_s(1)}|\geq |\mathcal{U}_{\sigma_s(2)}|\geq\dots\geq |\mathcal{U}_{\sigma_s(\Lambda)}|$.
\end{lemma}

The reader is referred to the extended version \cite[Section V]{PEU_arxiv_multi} for the proof of Lemma \ref{lem:cons_acyclic} which is an adaptation of \cite[Lemma 1]{WanTP15} to our setting.

\begin{remark}
The choice of the permutation $\sigma_s$ is critical for the development of a tight converse. Any other choice $\sigma\in S_\Lambda$ may result --- in some crucial cases --- in an acyclic subgraph with a smaller number of nodes and therefore a looser bound. This approach here deviates from the original approach in \cite[Lemma 1]{WanTP15}, which instead considered --- for each $\dv,\chi$, for the uniform user-to-cache association case of $K = \Lambda$ --- the set of \emph{all} possible permutations, that jointly resulted in a certain symmetry that is crucial to that proof. Here in our case, such symmetry would not serve the same purpose as it would dilute the non-uniformity in $\Lc$ that we are trying to capture. Our choice of a single carefully chosen permutation, allows for a bound which --- as it turns out --- is tight even in non-uniform cases.
\end{remark}
Having chosen an acyclic subgraph according to Lemma~\ref{lem:cons_acyclic}, we return to Lemma~\ref{cor_dof} and form the following lower bound by adding the sizes of all subfiles associated to the chosen acyclic graph as follows
\begin{equation}
T(\Uc,\boldsymbol{d},\chi)\geq T^{LB}(\Uc,\boldsymbol{d},\chi)
\end{equation} where
\begin{align}
&T^{LB}(\Uc,\boldsymbol{d},\chi) \defeq  \bigg( \sum_{j=1}^{\mathcal{L}_{1}}\sum_{\Tau_{1}\subseteq [\Lambda]\setminus \{\sigma_s(1)\}}|W^{\boldsymbol{d}_{\sigma_s(1)}(j)}_{\Tau_{1}}|\nonumber \\
&+ \sum_{j=1}^{\mathcal{L}_{2}}\sum_{\Tau_{2}\subseteq [\Lambda]\setminus \{\sigma_s(1),\sigma_s(2)\}}|W^{\boldsymbol{d}_{\sigma_s(2)}(j)}_{\Tau_{2}}|+\cdots \nonumber \\
&+ \sum_{j=1}^{\mathcal{L}_{\Lambda}}\sum_{\Tau_{\Lambda}\subseteq [\Lambda]\setminus \{\sigma_s(1),\dots,\sigma_s(\Lambda)\}}|W^{\boldsymbol{d}_{\sigma_s(\Lambda)}(j)}_{\Tau_{\Lambda}}|
\bigg). \label{eq:TLB}
\end{align}

Our interest lies in a lower bound for the worst-case delivery time/delay associated to profile $\Lc$. Such a worst-case naturally corresponds to the scenario when all users request different files, i.e., where all the entries of the demand vector $\dv(\mathcal{U})$ are different. The corresponding lower bound can be developed by averaging over worst-case demands. Recalling our set $\mathcal{D}_{\Lc}$, the worst-case delivery time can thus be written as
\begin{align}
T^*(\Lc)&\defeq \min_{\chi} \max_{(\mathcal{U},\dv) \in (\mathcal{U}_{\Lc},[N]^K)} T(\mathcal{U},\dv,\chi)\\
&\overset{(a)}{\geq} \min_{\chi} \frac{1}{|\mathcal{D}_{\Lc}|} \sum_{\dv(\mathcal{U}) \in \mathcal{D}_{\Lc}} T(\dv(\mathcal{U}),\chi)\label{eq:alternativedefinitionofT}
\end{align} where in step (a), we used the following change of notation $T(\dv(\mathcal{U}),\chi)\dfn T(\mathcal{U},\dv,\chi)$ and averaged over worst-case demands.

With a given class/profile $\Lc$ in mind, in order to construct $\mathcal{D}_{\Lc}$ (so that we can then average over it), we consider a demand vector $\dv\in \mathcal{D}_{wc}$ and a permutation $\pi\in S_{\Lambda}$. 
Then we create the following set of $\Lambda$ vectors
\begin{align*}
&\boldsymbol{d^{'}_1}= (d_1 : d_{\mathcal{L}_1}),\\
&\boldsymbol{d^{'}_2}= (d_{\mathcal{L}_1+1} : d_{\mathcal{L}_1+\mathcal{L}_2}),\\
& \vdots \\
&\boldsymbol{d^{'}_{\Lambda}}= (d_{\sum_{i=1}^{\Lambda-1}\mathcal{L}_{i}~+1} : d_{K}).
\end{align*}
For each permutation $\pi\in S_{\Lambda}$ applied to the set $\{1,2,\dots,\Lambda\}$, a demand vector $\dv(\mathcal{U})$ is constructed as follows
\begin{align}
\dv(\mathcal{U})&\dfn(\boldsymbol{d_1},\boldsymbol{d_2},\dots,\boldsymbol{d_\Lambda})\\
&=(\boldsymbol{d^{'}_{\pi^{-1}(1)}},\boldsymbol{d^{'}_{\pi^{-1}(2)}},\dots,\boldsymbol{d^{'}_{\pi^{-1}(\Lambda)}}).
\end{align}
This procedure is repeated for all $\Lambda!$ permutations ${\pi\in S_{\Lambda}}$ and all $P(N,K)$ worst-case demands $\dv\in \mathcal{D}_{wc}$. This implies that the cardinality of $\mathcal{D}_{\Lc}$ is ${|\mathcal{D}_{\Lc}|=P(N,K)\cdot \Lambda!}$.

Now the optimal worst-case delivery time in (\ref{eq:alternativedefinitionofT}) is bounded as
\begin{align}
T^{*}(\Lc) 
&= \min_{\chi}T(\Lc,\chi)\\
& \geq  \min_{\chi} \frac{1}{P(N,K)\Lambda!} \sum_{\dv(\mathcal{U}) \in \mathcal{D}_{\Lc}} T^{LB}(\dv(\mathcal{U}),\chi) \label{eq:lowerboundcompact}
\end{align}
where $T^{LB}(\dv(\mathcal{U}),\chi)$ is given by (\ref{eq:TLB}) for each reordered demand vector $\dv(\mathcal{U})\in \mathcal{D}_{\Lc}$. Rewriting the summation in (\ref{eq:lowerboundcompact}), we get
\begin{align}\label{eq:longinequality}
&\sum_{\dv(\mathcal{U})\in \mathcal{D}_{\Lc}}  T^{LB}(\dv(\mathcal{U}),\chi)= \nonumber \\
&\sum_{i=0}^{\Lambda}\sum_{n\in[N]}\sum_{\Tau\subseteq[\Lambda]:|\Tau|=i} |W^n_{\Tau}| \cdot \underbrace{\sum_{\dv(\mathcal{U})\in \mathcal{D}_{\Lc}} \mathds{1}_{\mathcal{V}_{\mathcal{J}_s^{\dv(\mathcal{U})}}}(W^n_{\Tau})}_{Q_{i}(W^n_\Tau)}
\end{align} where $\mathcal{V}_{\mathcal{J}_s^{\dv(\mathcal{U})}}$ is the set of vertices in the acyclic subgraph chosen according to Lemma \ref{lem:cons_acyclic} for a given $\boldsymbol{d}(\mathcal{U})$. In the above, $\mathds{1}_{\mathcal{V}_{\mathcal{J}_s^{\dv(\mathcal{U})}}}(W^n_{\Tau})$ denotes the indicator function which takes the value of 1 only if $W^n_{\Tau} \subset \mathcal{V}_{\mathcal{J}_s^{\dv(\mathcal{U})}}$, else it is set to zero.

A crucial step toward removing the dependence on $\Tau$, comes from the fact that
\begin{align}\label{eq:Qi}
Q_{i} &= Q_{i}(W^n_\Tau)\dfn \sum_{\dv(\mathcal{U})\in \mathcal{D}_{\Lc}} \mathds{1}_{\mathcal{V}_{\mathcal{J}_s^{\dv(\Uc)}}}(W^n_{\Tau}) \nonumber\\
=&{N-1 \choose K-1}\sum_{r=1}^{\Lambda}P(\Lambda-i-1,r-1)(\Lambda-r)!\mathcal{L}_{r} \nonumber\\
&\times P(K-1,\mathcal{L}_{r}-1) (K-\mathcal{L}_{r})! (\Lambda-i)
\end{align}
where we can see that the total number of times a specific subfile appears --- in the summation in \eqref{eq:longinequality}, over the set of all possible $\dv(\mathcal{U})  \in \mathcal{D}_{\Lc}$, and given our chosen permutation $\sigma_s$ 
--- is not dependent on the subfile itself but is dependent only on the number of caches $i=|\Tau|$ storing that subfile. The proof of \eqref{eq:Qi} can be found in the extended version \cite[Section V]{PEU_arxiv_multi} of our work.

In the spirit of~\cite{WanTP15}, defining
\begin{equation}
x_i\dfn\sum_{n\in[N]}\sum_{\Tau\subseteq[\Lambda]:|\Tau|=i}|W^n_{\Tau}|
\end{equation}
to be the total amount of data stored in exactly $i$ helper nodes, we see that
\begin{equation}\label{eq:sumfiles}
N=\sum_{i=0}^{\Lambda}x_i=\sum_{i=0}^{\Lambda}\sum_{n\in[N]}\sum_{\Tau\subseteq[\Lambda]:|\Tau|=i}|W^n_{\Tau}|
\end{equation}
and we see that combining \eqref{eq:lowerboundcompact}, \eqref{eq:longinequality} and \eqref{eq:Qi}, gives
\begin{equation}\label{eq:compacteq}
T(\Lc,\chi)\geq \sum_{i=0}^{\Lambda}\frac{Q_{i}}{P(N,K)\Lambda!}x_{i}.
\end{equation}
Now substituting \eqref{eq:Qi} into \eqref{eq:compacteq}, after some algebraic manipulations, we get that
\begin{align}
T(\Lc,\chi)&\geq \sum_{i=0}^{\Lambda}\frac{\sum_{r=1}^{\Lambda-i}\mathcal{L}_{r} {\Lambda-r\choose i}}{N{\Lambda\choose i}}x_{i} \label{eq:LBwithxi}\\
&=\sum_{i=0}^{\Lambda}\frac{x_{i}}{N}c_{i} \label{eq:LBwithxi_2}
\end{align}
where $c_{i}\triangleq \frac{\sum_{r=1}^{\Lambda-i}\mathcal{L}_r{\Lambda-r\choose i}}{{\Lambda\choose i}}$ decreases with $i\in \{0,1,\dots,\Lambda\}$.
The proof of the transition from \eqref{eq:compacteq} to \eqref{eq:LBwithxi}, as well as the monotonicity proof for the sequence $\{c_i\}_{i\in [\Lambda]\cup \{0\}}$, are given in the extended version of this work in \cite[Section V]{PEU_arxiv_multi}.

Under the file-size constraint given in \eqref{eq:sumfiles}, and given the following cache-size constraint
\begin{equation}
\sum_{i=0}^{\Lambda}i \cdot x_{i}\leq  \Lambda M \label{eq:constr2}
\end{equation}
the expression in~\eqref{eq:LBwithxi} serves as a lower bound on the delay of any caching-and-delivery scheme $\chi$ whose caching policy implies a set of $\{x_i\}$.

We then employ the Jensen's-inequality based technique of \cite[Proof of Lemma 2]{YuMA16} to minimize the expression in \eqref{eq:LBwithxi}, over all admissible $\{x_i\}$. Hence for any integer $\Lambda\gamma$, we have
\begin{equation}\label{eq:optimization1}
T(\Lc,\chi)\geq \frac{\sum_{r=1}^{\Lambda-\Lambda\gamma}\mathcal{L}_r{\Lambda-r\choose \Lambda\gamma}}{{\Lambda\choose \Lambda\gamma}}
\end{equation}
whereas for all other values of $\Lambda\gamma$, this is extended to its convex lower envelop.
The detailed derivation of~\eqref{eq:optimization1} can again be found in \cite[Section V]{PEU_arxiv_multi}.
This concludes lower bounding  $\max_{(\mathcal{U},\dv) \in (\mathcal{U}_{\Lc},[N]^K)} T(\mathcal{U},\dv,\chi)$, and thus --- given that the right hand side of \eqref{eq:optimization1} is independent of $\chi$ --- lower bounds the performance for any scheme $\chi$, which hence concludes the proof of the converse for Theorem~\ref{thm:PerClassSingleAntenna}.
\qed

\subsubsection{Proof of the converse for Corollary~\ref{cor:ressym} \label{sec:ConverseUniform}}
For the uniform case of $\Lc = [\frac{K}{\Lambda},\frac{K}{\Lambda},\dots,\frac{K}{\Lambda}]$, the lower bound in \eqref{eq:optimization1} becomes
\begin{align}  \frac{\sum_{r=1}^{\Lambda-\Lambda\gamma}\mathcal{L}_r{\Lambda-r\choose \Lambda\gamma}}{{\Lambda\choose \Lambda\gamma}}
& =\frac{K}{\Lambda}\frac{\sum_{r=1}^{\Lambda-\Lambda\gamma}{\Lambda-r\choose \Lambda\gamma}}{{\Lambda\choose \Lambda\gamma}} \\
  &\overset{(a)}{=} \frac{K}{\Lambda}\frac{{\Lambda \choose \Lambda\gamma+1}}{{\Lambda\choose \Lambda\gamma}}  \\
	& = \frac{K(1-\gamma)}{\Lambda\gamma+1}
  \end{align}
where the equality in step (a) is due to Pascal's triangle.  \qed


\section{Coded Caching Scheme \label{sec:scheme}}
This section is dedicated to the description of the placement-and-delivery scheme achieving the performance presented in Theorem~\ref{thm:PerClassSingleAntenna} and Corollary \ref{cor:ressym}. The formal description of the optimal scheme in the upcoming subsection will be followed by a clarifying example that demonstrates the main idea behind the design in Section \ref{subsec:example_scheme}.


\subsection{Description of the General Scheme}

\subsubsection{Cache Placement Phase}
The placement phase employs the original cache-placement algorithm of~\cite{MN14} corresponding to the scenario of having only $\Lambda$ users, each with their own cache. Hence --- recalling from~\cite{MN14} --- first each file $W^n$ is split into $\Lambda \choose \Lambda\gamma$ disjoint subfiles $W^n_\Tau$, for each $\Tau \subset [\Lambda]$, $|\Tau|=\Lambda\gamma$, and then each cache $\Zc_\lambda$ stores a fraction $\gamma$ of each file, as follows
\begin{equation}
\Zc_\lambda=\{W^n_\Tau :\Tau\ni\lambda,~ \forall n\in[N]\}.
\end{equation}

\subsubsection{Delivery Phase}
For the purpose of the scheme description only, we will assume without loss of generality that $|\mathcal{U}_1| \geq |\mathcal{U}_2| \geq \dots \geq |\mathcal{U}_{\Lambda}|$ (any other case can be handled by simple relabeling of the caches), and we will use the notation $\mathcal{L}_\lambda \defeq |\Uc_\lambda|$. Furthermore, in a slight abuse of notation, we will consider here each $\mathcal{U}_\lambda$ to be an \emph{ordered vector} describing the users associated to cache $\lambda$.

The delivery phase commences with the demand vector $\dv$ being revealed to the server.
Delivery will consist of $\mathcal{L}_1$ rounds, where each round $j\in[\mathcal{L}_1]$ serves users\footnote{A similar transmission method can be found also in the work of \cite{JinCaireGlobecom16} for the setting of decentralized coded caching with reduced subpacketization. } 
\begin{equation*}
\mathcal{R}_j=\bigcup_{\lambda\in[\Lambda]} \big( \mathcal{U}_{\lambda}(j):\mathcal{L}_\lambda \geq j \big)
\end{equation*}
and $\mathcal{U}_{\lambda}(j)$ is the $j$-th user in the set $\mathcal{U}_{\lambda}$.

\paragraph*{Transmission scheme}
For each round $j$, we create $\Lambda \choose \Lambda\gamma+1$ sets $\mathcal{Q}\subseteq [\Lambda]$ of size $|\mathcal{Q}|=\Lambda\gamma+1$, and for each set $\mathcal{Q}$, we pick the set of receiving users as follows
$$\chi_\mathcal{Q}=\bigcup_{\lambda\in \mathcal{Q}}\big( \mathcal{U}_{\lambda}(j):\mathcal{L}_\lambda \geq j \big).$$ Then if $\chi_\mathcal{Q}\neq \emptyset$, the server transmits the following message
\begin{equation} \label{eq:TransmissionSingleAntenna}
x_{\chi_{\mathcal{Q}}}=\!\!\!\!\bigoplus_{\lambda\in \mathcal{Q}:\mathcal{L}_\lambda \geq j} W^{d_{\mathcal{U}_{\lambda}(j)}}_{\mathcal{Q}\backslash{\{\lambda\}}}.
 \end{equation}
On the other hand, if $\chi_\mathcal{Q} = \emptyset$, there is no transmission.

\paragraph*{Decoding}
Directly from~\eqref{eq:TransmissionSingleAntenna}, we see that each receiver $\mathcal{U}_{\lambda}(j)$ obtains a received signal which takes the form
\begin{equation}
y_{\mathcal{U}_{\lambda}(j)}=W^{d_{\mathcal{U}_{\lambda}(j)}}_{\mathcal{Q}\backslash{\{\lambda\}}} + \underbrace{\bigoplus_{\lambda '\in \mathcal{Q}\backslash{\{\lambda\}}:\mathcal{L}_{\lambda'} \geq j} W^{d_{\mathcal{U}_{\lambda'}(j)}}_{\mathcal{Q}\backslash{\{\lambda'\}}}}_{interference} \label{eq:received_signal}
\end{equation}
which shows that the entire interference term experienced by receiver $\mathcal{U}_{\lambda}(j)$ can be `cached-out' because all the files $W^{d_{\mathcal{U}_{\lambda '}(j)}}_{\mathcal{Q}\backslash{\{\lambda '\}}}$ for all $\lambda^{'}\in \mathcal{Q}\setminus{\{\lambda\}}, \mathcal{L}_{\lambda'}\geq j$ that appear in this term, can be found --- since $\lambda\in \mathcal{Q}\backslash\{\lambda^{'}\}$ --- in cache $\lambda$ associated to this user.

This completes the description of the scheme.

\subsection{Calculation of the Delay}
To first calculate the delay needed to serve the users in $\mathcal{R}_j$ during round $j$, we recall that there are $\Lambda \choose \Lambda\gamma+1$ sets of users defined as $$\chi_\mathcal{Q}=\bigcup_{\lambda\in \mathcal{Q}}\big( \mathcal{U}_{\lambda}(j):\mathcal{L}_\lambda \geq j \big), \mathcal{Q}\subseteq [\Lambda]$$ and we recall that $|\mathcal{U}_1| \geq |\mathcal{U}_2| \geq \cdots \geq |\mathcal{U}_{\Lambda}|$. Furthermore we see that there are ${\Lambda-|\mathcal{R}_j| \choose \Lambda\gamma+1}$ such sets $\chi_\mathcal{Q}$ which are empty, which means that round $j$ consists of
\begin{equation}\label{eq:totsubround}
{\Lambda \choose \Lambda\gamma+1}-{\Lambda-|\mathcal{R}_j| \choose \Lambda\gamma+1}
\end{equation}
transmissions. Since each file is split into ${\Lambda\choose \Lambda\gamma}$ subfiles, the duration of each such transmission is
$\frac{1}{{\Lambda\choose \Lambda\gamma}}$.
Thus summing over all $\mathcal{L}_1$ rounds, the total delay takes the form
\begin{equation}\label{eq:totdelay1}
T=\sum_{j=1}^{\mathcal{L}_1}\frac{{{\Lambda \choose \Lambda\gamma+1}-{\Lambda-|\mathcal{R}_j| \choose \Lambda\gamma+1}}}{{\Lambda\choose \Lambda\gamma}}
\end{equation}
and after some basic algebraic manipulation (see~\cite[Section V]{PEU_arxiv_multi} for details), the delay takes the final form
\begin{equation}\label{eq:totdelay2}
T=\frac{\sum_{r=1}^{\Lambda-\Lambda\gamma}\mathcal{L}_r{\Lambda-r\choose \Lambda\gamma}}{{\Lambda\choose \Lambda\gamma}}
\end{equation} which concludes the achievability part of the proof. \qed

\subsection{Example for $K=N=8$, $\Lambda=4$ and $\Lc=(3,2,2,1)$\label{subsec:example_scheme}}
Consider a scenario with $K=8$ users $\{1,2,\dots,8\}$ assisted by $\Lambda=4$ helper caches, each of size $M=4$ units of file, storing content from a library of $N=8$ equally-sized files $W^1,W^2,\dots,W^8$.

In the cache placement phase, each file $W^n$ is first split into $6$ equally-sized subfiles $W^n_{1,2},W^n_{1,3},W^n_{1,4},W^n_{2,3},W^n_{2,4},W^n_{3,4}$, and then each cache $\lambda$ stores $W^n_{\Tau}:\Tau\ni\lambda,\forall n\in [8]$, where for example cache $1$ stores subfiles $W^n_{1,2},W^n_{1,3},W^n_{1,4}$.

In the subsequent cache assignment, users $\Uc_1=\{1,2,3\}$, $\Uc_2=\{4,5\}$, $\Uc_3=\{6,7\}$ and $\Uc_4=\{8\}$ are assigned to caches $1$, $2$, $3$ and $4$, respectively. As we see, this association has a profile of the form $\Lc=(3,2,2,1)$. We will assume without loss of generality a standard worst-case demand vector $\dv=(1,2,3,4,5,6,7,8)$.

Delivery takes place in $|\Uc_1|=3$ consecutive rounds, with each round respectively serving the following sets of users
\begin{align*}
&\mathcal{R}_{1}=\{1,4,6,8\}\\
&\mathcal{R}_{2}=\{2,5,7\}\\
&\mathcal{R}_{3}=\{3\}.
\end{align*}

In the first round, the server transmits the following messages
\begin{align*}
&x_{\{1,4,6\}}=W^1_{2,3}\oplus W^4_{1,3}\oplus W^6_{1,2}\\
&x_{\{1,4,8\}}=W^1_{2,4}\oplus W^4_{1,4}\oplus W^8_{1,2}\\
&x_{\{1,6,8\}}=W^1_{3,4}\oplus W^6_{1,4}\oplus W^8_{1,3}\\
&x_{\{4,6,8\}}=W^4_{3,4}\oplus W^6_{2,4}\oplus W^8_{2,3}
\end{align*}
and then decoding is done as in~\cite{MN14}. For instance, user $1$, upon receiving $x_{\{1,4,6\}}$, can decode its desired $W^1_{2,3}$ since it can fetch $W^4_{1,3}$ and $W^6_{1,2}$ at zero cost from its associated helper cache $1$. A similar procedure is applied by users $4$ and $6$ for the first transmission, as well as for the 3 subsequent XOR-messages.
In the second round, we have the following set of transmissions
\begin{align*}
x_{\{2,5,7\}}&=W^2_{2,3}\oplus W^5_{1,3}\oplus W^7_{1,2}\\
x_{\{2,5\}}&=W^2_{2,4}\oplus W^5_{1,4}\\
x_{\{2,7\}}&=W^2_{3,4}\oplus W^7_{1,4}\\
x_{\{5,7\}}&=W^5_{3,4}\oplus W^7_{2,4}\notag
\end{align*}
while in the last round, the server serves user $3$ with three consecutive unicast transmissions
$$x_{\{3\}}=W^3_{2,3}||W^3_{2,4}||W^3_{3,4}.
$$
Adding the delay for the above 11 transmissions, tells us that the overall normalized delivery time required to serve all the users is $T=\frac{11}{6}$. It is very easy to see that this delay remains the same --- given again worst-case demand vectors --- for any user-to-cache association $\mathcal{U}$ with the same profile $\Lc=(3,2,2,1)$. Every time, this delay matches the converse $$T^*(3,2,2,1)\geq \frac{\sum_{r=1}^{2}\mathcal{L}_r {4-r \choose 2}}{{4 \choose 2}}=\frac{11}{6}$$ of Theorem \ref{thm:PerClassSingleAntenna}\footnote{Note that optimality is maintained despite the fact that the server has transmitted seemingly inefficient messages that involved less than $\Lambda\gamma+1=3$ users.}. 

\section{Conclusions and final remarks\label{sec:discussion}}
The work is among the first to employ index coding as a means of providing (in this case, exact) outer bounds for more involved cache-aided network topologies that better capture aspects of larger cache-aided networks, such as having shared caches and a variety of user-to-cache association profiles. Dealing with such non uniform profiles, raises interesting challenges in redesigning converse bounds as well as redesigning coded caching which is known to generally thrive on symmetry. The result can also be useful in providing guiding principles on how to assign shared caches to different users. Finally we believe that the adaptation of the outer bound technique to non-uniform settings may also be useful in analyzing different applications like distributed computing \cite{LiAliAvestimehrComputIT18,PLE:18a,KonstantinidisRamamoorthyArxiv18,YanYangWiggerArxiv18,MingyueJiISIT18} or data shuffling \cite{AttiaTandon16,AttiaTandonISIT18,WanTuninettiShuffling18,MohajerISIT18} which can naturally entail such non uniformities.

\bibliography{final_refs2}
\bibliographystyle{IEEEtran}
\end{document}